\RequirePackage{etex}
\documentclass[11pt]{article}
\usepackage[margin=1in]{geometry}
\usepackage{amsmath,amssymb,amsthm,mathtools}
\usepackage{graphicx}
\usepackage{tikz}
\usepackage{physics}
\usepackage{bm}
\usepackage{microtype}
\usepackage[dvipsnames]{xcolor}
\usepackage[hypertexnames=false,colorlinks=true,linkcolor=MidnightBlue,citecolor=ForestGreen,urlcolor=BrickRed]{hyperref}
\usepackage[nameinlink]{cleveref}
\usepackage{autonum}

\numberwithin{equation}{section}
\newtheorem{theorem}{Theorem}
\newtheorem{lemma}{Lemma}
\newtheorem{proposition}[lemma]{Proposition}
\newtheorem{corollary}[lemma]{Corollary}

\DeclarePairedDelimiter{\probbracket}{[}{]}
\DeclarePairedDelimiter{\expectbracket}{[}{]}

\makeatletter

\newdimen\softdelimiterlimit
\NewDocumentCommand{\softbracket}{o m}{%
  \IfValueTF{#1}
    {\mathpalette\softbracket@conditional{{#2}{#1}}}
    {\mathpalette\softbracket@plain{#2}}%
}

\newcommand{\softbracket@plain}[2]{%
  \setbox\z@=\hbox{$\m@th#1#2$}%
  \dimen@=\ht\z@
  \advance\dimen@\dp\z@
  \ifdim\dimen@>\softdelimiterlimit
    \biggl[#2\biggr]%
  \else
    \left[#2\right]%
  \fi
}

\newcommand{\softbracket@conditional}[2]{%
  \softbracket@conditional@aux#2{#1}%
}

\newcommand{\softbracket@conditional@aux}[3]{%
  \setbox\z@=\hbox{$\m@th#3#1\,\vert\,#2$}%
  \dimen@=\ht\z@
  \advance\dimen@\dp\z@
  \ifdim\dimen@>\softdelimiterlimit
    \biggl[#1\,\biggm\vert\,#2\biggr]%
  \else
    \left[#1\,\middle\vert\,#2\right]%
  \fi
}

\makeatother

\NewDocumentCommand{\prob}{s o m}{%
  \mathbb{P}%
  \IfBooleanTF{#1}
    {%
      \IfValueTF{#2}
        {\softbracket[#2]{#3}}
        {\softbracket{#3}}%
    }
    {\probbracket{#3\IfValueT{#2}{\,\vert\,#2}}}%
}

\NewDocumentCommand{\EX}{s o o m}{%
  \mathbb{E}\IfValueT{#2}{_{#2}}%
  \IfBooleanTF{#1}
    {%
      \IfValueTF{#3}
        {\softbracket[#3]{#4}}
        {\softbracket{#4}}%
    }
    {\expectbracket{#4\IfValueT{#3}{\,\vert\,#3}}}%
}

\DeclareMathOperator{\Ad}{Ad}
\newcommand{\lowT}{\raisebox{-0.5ex}{$\scriptstyle T$}}
\newcommand{\poly}{\operatorname{poly}}
\newcommand{\Tact}{T_{\mathrm{act}}}

\title{\textbf{One Gate at a Time: Complexity Growth in Random Quantum Circuits}}
\author{Zhi Li\\{\normalsize IBM Research}}
\date{}

\begin{document}
\maketitle

\begin{abstract}
    A random unitary quantum circuit is expected to be incompressible for exponentially long times.
    We show that the constant-error circuit complexity of a random unitary circuit grows almost linearly with time as $\Omega(T/\log T)$. 
    The bound holds for all $2\leq T\leq 4^n$ where $n$ is the system size, and involves no other $n$-dependence.
    This improves previous lower bounds derived from spectral gaps and unitary designs by a factor of $\poly(n)$.
    Drawing on insights from stochastic calculus, geometric functional analysis, and randomized linear algebra, our approach exploits the circuit's response to variations of individual gates and requires no control over convergence to high-order unitary designs.
  \end{abstract}

\section{Introduction}

Circuit complexity quantifies the minimum number of elementary gates needed to synthesize a unitary and is therefore a basic measure of computational resources in quantum information. 
Most unitaries on $n$ qubits have exponentially large complexity due to a counting argument \cite{BarencoEtAl1995,Knill1995}, 
and a generic time evolution is expected to enter the high-complexity region eventually after a long period of evolution \cite{RobertsYoshida2017,BrownSusskind2018,BCHKP2021,oszmaniec2024saturation}. 
A natural question is how quickly this happens. 
In chaotic many-body systems, circuit complexity is expected to grow linearly with time, continuing long after local observables have equilibrated and entanglement has saturated, before reaching its maximum only at times exponential in $n$ \cite{BrownSusskind2018}.
Random quantum circuits offer a concrete setting in which to test this expectation: does a typical random sequence of local gates admit a substantially shorter implementation?

This seemingly simple question has attracted considerable attention.
For \emph{exact} complexity, where no approximation error is allowed, Haferkamp et al. proved an $\Omega(T/\mathrm{poly}(n))$ lower bound up to exponentially long times \cite{HFKEY2022}. See also the two short proofs by Li \cite{Li2022}.

However, the exactness condition makes this result fragile to perturbations, whereas practical quantum information tasks generally allow some error.
In this regard, the question is substantially more difficult.
For a long time, people have only been able to prove a polynomial lower bound $\Omega(\poly(T)/\poly(n))$ \cite{BHH2016,BCHKP2021,Haferkamp2022} (where the $\poly(T)$ has degree less than 1, e.g., $0.2-o(1)$ in \cite{Haferkamp2022}) for constant approximation error, or an $\Omega(T/\poly(n))$ lower bound for exponentially small approximation error and a non-universal gate set \cite{Haferkamp2023}. 
A robust linear growth result was finally obtained in \cite{CHHLMT}, where the authors established an $\Omega(T/\poly(n))$ lower bound for $T\lesssim 2^{n/2}$.
The common route taken by these results is through the spectral gap of the associated random walk on the unitary group.
In particular, \cite{CHHLMT} establishes a lower bound on the $k$-th moment spectral gap that is uniform in $k$ for $k\lesssim 2^{n/2}$.
The spectral gap result is improved in \cite{BaerHaah2026}, based on which one can improve the degree of the $\poly(n)$ in the denominator and push the time range all the way to $T\sim 4^{n}$. 

These results on spectral gaps and, closely related, on unitary designs \cite{emerson2003pseudo,DankertCleveEmersonLivine2009} are powerful and have many other implications in quantum information.
For the specific purpose of studying complexity growth, however, a more direct argument may offer some conceptual and quantitative advantages.
Conceptually, complexity is expected to grow because each step explores a new direction in the vast unitary group, so collisions are rare \cite{susskind2020three}.
This intuition also underlies the proofs in \cite{HFKEY2022,Li2022}, which show that each gate (or each layer in a fixed architecture) increases the dimension of the subspace of $\mathbf U(2^n)$ being explored.
In the spectral-gap and unitary-design approach, this geometric mechanism is less explicit.
Quantitatively, the standard conversion from a spectral gap to unitary design \cite{BHH2016} may incur an additional factor of $n$ (in light of shallow but more structured unitary designs \cite{SchusterHaferkampHuang2025,laracuente2026approximate}), resulting in a corresponding loss in the complexity lower bound even when the spectral-gap estimate is optimal.
These considerations motivate us to seek a more direct route to complexity growth.

\subsection*{Main results}

This work develops a geometric explanation for studying complexity growth in random quantum circuits.
Our approach formalizes the geometric intuition that successive random gates typically explore new directions in the unitary group.
In particular, our proof does not use information about higher-level spectral gaps or unitary designs.

For concreteness, we focus on the all-to-all 2-local random unitary circuit, where each gate is drawn Haar randomly from $\mathbf{U}(4)$, acting on a random pair of qubits. 
Denote the total number of gates by $T$ and the final random unitary by $U_T$.
Our main technical result is the following \emph{setwise} small-ball\footnote{In probability theory, a natural terminology for it would be ``anti-concentration''. However, in the quantum information literature it usually has a different meaning \cite{DHB2022}, so we avoid this terminology.} estimate (\cref{sec:all-to-all-local}).

We prove that for every $T\le D$ and every nonempty finite subset $\mathcal V\subseteq\mathbf U(2^n)$ such that $\log\abs{\mathcal V}\le c_0 T$ (where $c_0>0$ is an absolute constant),
\begin{equation}\label{eq:mainsmall1}
  \prob*{\min_{V\in\mathcal V}d_2(U_T,V)\le1}
  \le \exp\left(-\frac{T}{\poly(n)}\right),
  % +\abs{\mathcal V}e^{-\Omega(T)},
\end{equation}
where $d_2$ is the normalized Hilbert-Schmidt distance.
% when $T=O(n^2)$, the first term can be improved to $e^{-\Omega(T)}$. 
Applying this estimate to an $\epsilon$-net of quantum circuits together with a more direct proof for short time complexity growth, 
we establish the following almost-linear bound in the full pre-saturation regime (\cref{sec:complexity}): for an absolute constant $c>0$, and every $2\le T\le D$,
\begin{equation}\label{eq:comp-main}
  \prob*{\mathcal C_{1/2}(U_T)\le c\frac{T}{\log T}}
  \le \exp\left(-\frac{T}{\poly(n)}\right).
\end{equation}
% For $T=O(n^2)$, the failure probability again improves to $e^{-\Omega(T)}$. 
Namely, the constant-error complexity $\mathcal C_{1/2}(U_T)$ grows as $\Omega(T/\log T)$ with no explicit dependence on $n$, all the way to times of order $4^n$, the saturation time scale.

While the main focus of this paper is all-to-all 2-local circuits, our proof strategy works more cleanly in a different random-unitary model known as random Pauli rotations \cite{HaahLiuTan2025}.
Here, each step of the random walk applies $\exp(i\theta P)$, where $P$ is a random Pauli operator and $\theta$ is a random angle.
For pedagogical reasons, we consider random Pauli rotations first (\cref{sec:rpr}), where we establish an easy version of the small-ball estimate for $T=O(2^n)$ in this model:
\begin{equation}\label{eq:mainsmall2}
  \sup_{V\in \mathbf U(N)}
  \prob*{d_2(U(\theta),V)\le 1}
  \le \exp(-\Omega(T)).
\end{equation}

Although our complexity-growth result does not use a uniform spectral gap or fast convergence to high-moment designs, the random circuits we considered do have such properties.
To demonstrate a separation, we consider Kac's random walk on $\mathbf{SO}(2^n)$ \cite{Kac1956} (\cref{sec:Kac}). This walk is known to have a spectral gap $\Theta(2^{-n})$ \cite{CarlenCarvalhoLoss2003}, so the spectral-gap method does not directly yield a growing complexity lower bound for $T=O(2^n)$. Nevertheless, we show that the complexity grows as $\Omega(T/\log(nT))$ throughout this exponentially large time range.

\subsection*{Compare with previous results}

While the primary aim of this work is conceptual, we compare our bounds with those obtained from spectral gaps and unitary designs for completeness.

For all-to-all 2-local random quantum circuits, \cite{BaerHaah2026} shows that the uniform spectral gap satisfies $\gamma=\Omega(1/n)$, which is optimal in its dependence on $n$ \cite{mittal2023local}. 
Using the standard spectral gap to complexity conversion \cite{BHH2016, BCHKP2021,CHHLMT}, we can get the following bound on complexity growth for $n \leq T\leq 4^n$ (see \cref{sec:spectral-gap-complexity} for details):
\begin{align}
  \prob*{\mathcal C_{1/2}(U_T)\le
  \frac{cT}{n\log(T)}} 
  \le \exp(-\Omega(T/n)).
\end{align}
Compared to \cref{eq:comp-main}, the complexity lower bound above contains an extra factor $n$ in the denominator.

\subsection*{Reading guide}

The proofs in \cref{sec:rpr,sec:all-to-all-local} can be read independently. 
In \cref{sec:intuition}, we explain the core mechanism underlying the proofs. We recommend reading this subsection before the full proofs. 
Note, however, that \cref{sec:intuition} is intended to explain the underlying mechanism of the proofs and how they were naturally discovered. The full proofs do not follow every step outlined there exactly.

\Cref{sec:rpr} presents the proof strategy in a streamlined setting, and some aspects of the theorems (such as the time range) are intentionally not optimized. 
Likewise, \cref{sec:Kac} is intended to clarify the separation rather than to obtain optimal bounds. 
Most numerical constants throughout the paper are not optimized.

\section{Technical overview}\label{sec:overview}

\subsection{Notation}\label{sec:notation}
Throughout this paper, we consider quantum systems with $n$ qubits.
Write $N=2^n$.
We write $\mathbf U(N)$ for the unitary group on $(\mathbb C^2)^{\otimes n}$. 
We view $\operatorname{Mat}(N,\mathbb C)$ as a Hilbert space by defining:
\begin{equation}
  \braket{A}{B}=\tau(A^\dagger B)=\frac1N\tr A^\dagger B,
  \qquad
  \norm{A}_2 = \braket{A}{A}^{1/2}.
\end{equation}
(Throughout this paper, Dirac notation is used for elements of $\operatorname{Mat}(N,\mathbb C)$.)
The group $\mathbf U(N)$ acts unitarily on $\operatorname{Mat}(N,\mathbb C)$ by conjugation:
\begin{equation}
  \Ad_g(A) = gAg^{\dagger},\qquad\text{where~}
  g\in \mathbf U(N), ~
  A\in \operatorname{Mat}(N,\mathbb C).
\end{equation}
For $U, V \in \mathbf U(N)$, define the normalized Hilbert-Schmidt distance:
\begin{equation}
  d_2(U,V)=\norm{U-V}_2.
\end{equation}
Note that 
\begin{equation}
    d_2(U,V)^2=2-2\Re\tau(V^\dagger U).
\end{equation}

We define 
\begin{equation}
  \mathfrak g_n=\{A \mid A\in \operatorname{Mat}(N,\mathbb C), ~A=A^\dagger, ~\tr A=0\},
  \qquad
  \dim_{\mathbb R} \mathfrak g_n =D :=N^2-1.
\end{equation}
Equipped with the inner product inherited from $\operatorname{Mat}(N,\mathbb C)$, $\mathfrak g_n$ is a real Hilbert space.
We will occasionally consider its complexification without causing confusion, naturally identified with the complex Hilbert subspace of traceless matrices in $\operatorname{Mat}(N,\mathbb C)$.

We use $\Tr$ to denote the trace of a linear operator on $\operatorname{Mat}(N,\mathbb C)$, $\mathfrak g_n$, or its complexification; in contrast, $\tr$ denotes the trace on $(\mathbb C^2)^{\otimes n}$.

\paragraph{Circuit complexity}
For $U\in\mathbf U(N)$, let $\mathcal C_\epsilon(U)$ be the least number of arbitrary two-qubit gates, acting on any pairs of qubits, needed to approximate $U$ within operator-norm error $\epsilon$.
Note that
\begin{equation}
      d_2(U,V) \le \norm{U-V}_{\mathrm{op}}.
\end{equation}

\paragraph{Random Pauli rotations} Let
\begin{equation}
  \mathfrak p_n^\circ
  =\{I,X,Y,Z\}^{\otimes n}\setminus\{I^{\otimes n}\}, \qquad \abs{\mathfrak p_n^\circ}=D,
\end{equation}
be the set of nonidentity Pauli operators.
The nonidentity Pauli operators form an orthonormal basis of $\mathfrak g_n$,
\begin{equation}\label{eq:orthonormalP}
  \frac{1}{D}\sum_{P\in\mathfrak p_n^\circ}\ketbra{P}
  = \frac{1}{D}\mathsf I_{\mathfrak g_n}.
\end{equation}

We consider a random process of $T$ steps.
For $j=1,\ldots,T$, choose $P_j$ uniformly from $\mathfrak p_n^\circ$ and $\theta_j$ uniformly from $\mathbb T=\mathbb R/(2\pi\mathbb Z)$, independently. 
The corresponding one-step random Pauli rotation is $e^{i\theta_jP_j}$, and the product of $T$ such rotations is
\begin{equation}\label{eq:Utheta}
  U(\theta) = e^{i\theta_TP_T}\cdots e^{i\theta_1P_1},
  \qquad \theta=(\theta_1,\ldots,\theta_T).
\end{equation}

\paragraph{All-to-all 2-local random unitaries}
For $j=1,\ldots,T$, choose an unordered pair $e_j\in\binom{[n]}2$ uniformly and, conditional on $e_j$, choose $G_j$ Haar-randomly from the embedded subgroup $\mathbf U(4)_{e_j}$ acting on the two qubits in $e_j$. The product of $T$ such gates is
\begin{equation}
  U_T=G_T\cdots G_1.
\end{equation}

\subsection{Core intuition}\label{sec:intuition}

The main results of this work are the small-ball estimates \cref{eq:mainsmall1,eq:mainsmall2}.
The results on complexity growth follow directly by an $\epsilon$-net argument.
In this subsection, we provide a heuristic explanation of the small-ball estimate.

\paragraph{Random Pauli rotations} 

While the model is defined by sequential application of a time-independent random walk, we can regard the final distribution as the pushforward distribution of the Haar distribution on the torus $\mathbb{T}^T$ through the function $U(\theta)$. 

We compute the partial derivative $\pdv{U}{\theta_j}$ by pushing $e^{i\theta_j P_j}$ to the rightmost:
\begin{equation}\label{eq:partiald}
  \partial_j U(\theta) := \pdv{U}{\theta_j} = i U(\theta) X_j(\theta),
\end{equation}
where
\begin{equation}\label{eq:adjX}
  \qquad X_j(\theta) = U_{<j}(\theta)^\dagger P_j U_{<j}(\theta) = \Ad_{U_{<j}(\theta)^\dagger}(P_j),
  \qquad U_{<j}(\theta) = e^{i\theta_{j-1}P_{j-1}}\cdots e^{i\theta_1P_1}.
\end{equation}
Note that $X_j(\theta)$ is independent of $\theta_j$, so we further have:
\begin{equation}
    % \pdv[2]{U(\theta)}{\theta_j} = -U(\theta),
    \partial_j^2 U(\theta) = -U(\theta).
\end{equation}
Let us consider the overlap of $U(\theta)$ with an arbitrarily fixed unitary $V$:
\begin{equation}\label{eq:deff}
  f(\theta)=\Re\tau\left(V^\dagger U(\theta)\right) = 1-\frac12 d_2(U(\theta),V)^2.
\end{equation}

Now assume $\theta$ is subjected to the standard $T$-dimensional Brownian motion on $\mathbb{T}^T$. Itô's formula \cite{ito1951stochastic,karatzas1991brownian} then gives:
\begin{equation}\label{eq:defdynamics}
  \dd f = -\frac{T}{2} f \dd t -\Im\tau\left(V^\dagger U(\theta)X_j(\theta)\right)\cdot \dd B_t^{(j)}.
\end{equation}

This resembles the one-dimensional Ornstein–Uhlenbeck process \cite{UhlenbeckOrnstein1930}:
\begin{equation}\label{eq:OU}
\dd f_t = -\kappa f_t \dd t + \sigma \dd B_t.
\end{equation}
Here $\kappa$ is the mean-reversion rate, and $\sigma$ represents the noise volatility. 
The width of the stationary distribution is determined by the relative strength between the noise and the mean-reversion rate: it is well known that \cref{eq:OU} has the stationary Gaussian distribution $\mathcal{N}(0,\frac{\sigma^2}{2\kappa})$.
By analogy (the noise term in \cref{eq:defdynamics} is not a standard Brownian motion because it depends on $\theta$), our dynamics \cref{eq:defdynamics} has a mean-reversion rate $\kappa\sim T$, and
\begin{equation}
  \sigma^2\lesssim  \sum_j \abs{\tau(X_j(\theta) V^\dagger U(\theta))}^2 
  = \bra{V^\dagger U(\theta)}\Big(\sum_j\ketbra{X_j(\theta)}\Big)\ket{V^\dagger U(\theta)}.
\end{equation} 

The operator in the middle is called the frame operator for vectors $\ket{X_j}$.
To bound $\sigma$, one would ideally like the frame operator to satisfy:
\begin{equation}\label{eq:KT1}
  \sum_j\ketbra{X_j(\theta)} \overset{?}{\preceq} B \mathsf I_{\mathfrak g_n},
\end{equation}   
with an absolute constant $B=O(1)$.
Intuitively, it says that the vectors $\ket{X_j}$ cannot concentrate too strongly along any single direction; for instance, one can take $B=1$ when they are orthonormal.
Such a bound is generally too strong. Instead, using the randomness of $\{P_j\}$, in particular \cref{eq:orthonormalP}, we will select sufficiently many (a constant portion of $T$) of these vectors whose frame operator obeys the required uniform bound.

Under such a condition, the volatility $\sigma$ will be $O(1)$.
It is then natural to expect that the stationary distribution of our $f$ in \cref{eq:deff} is peaked at 0, with a width $O(1/\sqrt T)$, which would imply the desired small-ball estimate:
\begin{equation}\label{eq:anti1}
  \prob*{d_2(U(\theta),V)\leq 1} = \prob*{f(\theta)> \frac12} \leq \exp(-\Omega(T)).
\end{equation}

\paragraph{All-to-all 2-local random unitaries}

Most of the preceding argument also applies to all-to-all 2-local random unitary circuits.
The main difficulty is that the tangent directions generated by 2-local gates are less random than those generated by random Pauli operators.
More precisely, for random Pauli operators, \cref{eq:orthonormalP} gives, at each step, an average operator with norm $1/D$.
By contrast, a 2-local gate provides only $\Theta(n^2)$ possible directions at each step, whose average has operator norm $\Theta(1/n^2)$.

Fortunately, the actual tangent direction also involves the adjoint action of a random unitary; see \cref{eq:adjX}.
We will essentially show that further averaging over these adjoint actions reduces the problematic $\Theta(1/n^2)$ scale to $\Theta(1/D)$, allowing the argument to proceed.

\subsection{Connection with other fields}

The Ornstein--Uhlenbeck process motivates the balance between mean reversion and noise in our argument.
The proof itself, however, uses integration by parts rather than explicit stochastic processes.
This is largely a matter of bookkeeping: one could instead follow the stochastic-process intuition more directly using martingale concentration inequalities.
In fact, the differential operator used in the proof is the formal generator of the underlying auxiliary stochastic process.
Using generator identities to study invariant distributions, approximate probability laws, and derive concentration inequalities is a standard approach \cite{ledoux2001concentration, BakryGentilLedoux2014}, appearing, for example, in Stein's method \cite{Ross2011,Chatterjee2007}.

The selection of vectors is a recurring topic in geometric functional analysis.
Selecting a large subset for which a bound of the form \cref{eq:KT1} holds is a Kashin--Tzafriri type problem \cite{KashinTzafriri1993}, and the construction in \cref{sec:greedy} is inspired by the proof of the Dvoretzky--Rogers lemma \cite{DvoretzkyRogers1950}.
For our purpose, we need to ensure that our selector is adaptive, which connects our argument with online methods in randomized linear algebra.
The exponential-moment estimate used to analyze the first-stage selector in \cref{sec:all-to-all-local} is related to the approach to matrix concentration for expander walks \cite{GLSS2018}.
The analysis of the second-stage selector is closely related to that of online row sampling \cite{CohenMuscoPachocki2020}.

\section{Random Pauli rotations}\label{sec:rpr}

In this section, we prove a small-ball estimate for random Pauli rotations. 
We fix $T$ and let $U(\theta)$ be the result after $T$ steps of the random Pauli rotation process.

\begin{theorem}\label{cor:rpr-anticoncentration}
For all $1\le T\le 2^{n-4}$,
\begin{equation}
  \sup_{V\in \mathbf U(N)}
  \prob*{d_2(U(\theta),V)\le 1}
  \le \exp(-\Omega(T)).
\end{equation}
\end{theorem}

We prove this theorem by formalizing the intuition outlined in \cref{sec:intuition}.
The first lemma establishes a small-ball estimate under \cref{eq:KT2}, the appropriate replacement for \cref{eq:KT1}.
The second lemma shows that, with high probability, we can select at least $\Theta(T)$ directions while ensuring that \cref{eq:KT2} holds.

\subsection{Small-ball estimates via bounded frame operators}

Let us fix Pauli operators $P_1,\ldots,P_T$ for now. In this subsection only, we use $\prob{\cdot}$ and $\EX{\cdot}$ for the standard Haar probability and expectation on $\mathbb T^T$.
For $j\in[T]$, let $a_j:\mathbb T^T\to\{0,1\}$ be a function independent of $\theta_j$. Define
\begin{align}
  \mathsf S(\theta)=\sum_{j=1}^T a_j(\theta)\ketbra{X_j(\theta)},
\end{align}
where $X_j$ is as defined in \cref{eq:adjX}.
Also define
\begin{equation}
  K(\theta)= \Tr \mathsf S(\theta) = \sum_{j=1}^T a_j(\theta) \in \mathbb N.
\end{equation}

\begin{lemma}\label{lem:anticoncentration}
Suppose that, for some $B>0$,
\begin{equation}\label{eq:KT2}
  \mathsf S(\theta)\preceq B \mathsf I_{\mathfrak g_n},
  \qquad \forall\theta\in\mathbb T^T.
\end{equation}
Also suppose that $\partial_j a_j=0$ $(\forall j\in[T])$.
Then, for every unitary $V\in \mathbf U(N)$,
\begin{equation}
  \prob{d_2(U(\theta),V)\le1}
  \le
  \prob{K(\theta)<\frac{T}{2}}
  +\exp(-\frac{T}{16B}).
\end{equation}
\end{lemma}

\begin{proof}
  Fix the unitary $V$. Set $f(\theta)=\Re\tau(U(\theta)V^\dagger)$.
  A straightforward calculation shows
\begin{equation}\label{eq:fupdate}
  f(\theta+h \mathbf e_j)
  =\cos h\,f(\theta)+\sin h\,W_j(\theta),
\end{equation}
where $W_j(\theta)=-\Im \tau\qty(U(\theta)X_j(\theta) V^\dagger )$, and $\mathbf e_j$ is the unit vector in the $j$-th direction of $\mathbb{T}^T$.
The assumption \cref{eq:KT2} then implies that
\begin{equation}\label{eq:boundW}
  \begin{aligned}
    \sum_{j} a_j W_j(\theta)^2
    \leq \sum_{j:\,a_j(\theta)=1}
      \abs{\tau(X_j(\theta)V^\dagger U(\theta))}^2
    \leq B\norm{V^\dagger U(\theta)}_2^2
      =B.
  \end{aligned}
\end{equation}

It follows from \cref{eq:fupdate} that:
\begin{equation}
  \pdv{f(\theta)}{\theta_j} = W_j(\theta),\qquad \pdv[2]{f(\theta)}{\theta_j}=-f(\theta).
\end{equation}
Therefore, for any function $\phi:\mathbb R\to \mathbb R$, we have:
\begin{equation}
  \EX{a_j(\phi'(f)W_j^2-\phi(f)f)}
  =\EX{\pdv{}{\theta_j}\bigl(a_j\phi(f)W_j\bigr)}=0.
\end{equation}

Summing over $j$, we get:
\begin{equation}\label{eq:bypart}
  \EX{K\phi(f)f} = \EX{\phi'(f)\sum_j a_j W_j^2}.
\end{equation}

Recall that our goal is to understand $\prob{f\geq t}=\EX{\mathbf1_{\{f\geq t\}}}$.
We denote it by $H(t)$.
We (formally\footnote{One may use a family of smooth approximations for mathematical rigor.}) take $\phi(x)=\mathbf 1_{\{x\ge t\}}$ where $t\geq 0$; then \cref{eq:bypart} implies
\begin{equation}
\EX{Kf\,\mathbf1_{\{f\geq t\}}} = \EX{\delta(f-t)\sum_j a_j W_j^2}\leq B\EX{\delta(f-t)}=-BH'(t).
\end{equation}
On the other hand, by restricting to the complement of $E_{\mathrm{bad}}=\{\theta:K(\theta)<T/2\}$, we get
\begin{equation}
  \EX{Kf\,\mathbf1_{\{f\geq t\}}} 
  \geq \frac{Tt}{2} \prob*{\{f\geq t, K(\theta)\geq \frac{T}2\}}
  \geq \frac{Tt}{2} (H(t)-\prob{E_{\text{bad}}}).
\end{equation}
Therefore,
\begin{equation}
  \frac{\dd}{\dd t}(H(t)-\prob{E_{\text{bad}}})
  \leq  -\frac{Tt}{2B} (H(t)-\prob{E_{\text{bad}}})
\end{equation}
Grönwall's inequality then implies that (the r.h.s. below could be negative):
\begin{equation}
  H(t)-\prob{E_{\text{bad}}} \leq \exp(-\frac{Tt^2}{4B})(H(0)-\prob{E_{\text{bad}}}).
\end{equation}
Picking $t=1/2$ and noting that $H(0)\leq 1$, we find
\begin{equation}
  \prob*{f\ge \frac12} \leq \prob{E_{\text{bad}}} + \exp(-\frac{T}{16B}).
\end{equation}
\end{proof}

\subsection{Greedy selection}\label{sec:greedy}

In the above lemma, the Pauli operators $P_1,\cdots P_T$ were fixed, and a small-ball estimate was shown assuming the existence of $\{a_j\}$ such that \cref{eq:KT2} holds.
% In the following, we show that with high probability over the random choices of Pauli operators, such $\{a_j\}$ exists.
In the following, we construct such $\{a_j\}$ and show that, with high probability (over random Pauli operators), at least $T/2$ directions are selected.

Given $P_1,\ldots,P_T$ and $\theta\in\mathbb T^T$, we use the following greedy algorithm to define $a_j(\theta)$ (equivalently, to select a subset of $\{X_j\mid j\in[T]\}$).
We start from a subspace $\mathcal W_0=\{0\}$.
For $j=1,\ldots,T$, we iteratively define
\begin{equation}
  \begin{aligned}
      a_j
  &=\mathbf 1\qty{
    \norm{\operatorname{Proj}_{\mathcal W_{j-1}}X_j(\theta)}_2^2
    \le\frac1{32T}  },\\
  \qquad
  \mathcal W_j
  &=\mathcal W_{j-1}+\operatorname{span}\{a_j X_j(\theta)\}.
  \end{aligned}
\end{equation}
Namely, we assign $a_j=1$ and add $X_j$ to the selected set if $X_j$ is almost orthogonal to the span of previously selected vectors.

It is clear that each $a_j$, as a function of $\theta$, is independent of $\theta_j$.

\begin{lemma}\label{prop:selector}
For every Pauli sequence and every $\theta\in\mathbb T^T$,
\begin{equation}\label{eq:framebound}
  \mathsf S(\theta)=\sum_{j=1}^T a_j(\theta)\ketbra{X_j(\theta)}
  \preceq2\mathsf I_{\mathfrak g_n}.
\end{equation}
Moreover, for $1\le T\le 2^{n-4}$, 
\begin{equation}\label{eq:rprKbound}
  \prob*{\sum_{j=1}^T a_j < \frac{T}{2}}   \le e^{-\Omega(T)}.
\end{equation}
\end{lemma}

\begin{proof}

For an accepted $j$, denote the normalized component of $X_j$ orthogonal to $\mathcal W_{j-1}$ as $Y_j$. The vectors $Y_j$ are orthonormal, and
\begin{equation}
  \norm{X_j(\theta)-Y_j}_2^2
  \le 2\norm{\operatorname{Proj}_{\mathcal W_{j-1}}X_j(\theta)}_2^2
  \le\frac1{16T}.
\end{equation}
Let $\mathrm X$ and $\mathrm Y$ be the matrices with columns $X_j$ and $Y_j$ over accepted indices; they may be viewed as linear maps $\mathbb R^K\to \mathfrak g_n$.
We have
\begin{equation}
  \norm{\mathrm X-\mathrm Y}_{\mathrm{op}}\leq
  % \norm{\mathrm X-\mathrm Y}_{\mathrm F}
  \qty(\sum a_j \norm{X_j(\theta)-Y_j}_2^2)^{1/2}
  \le \frac{1}{4},
  \qquad
  \norm{\mathrm Y}_{\mathrm{op}}\le1.
\end{equation}
Therefore, $  \norm{\mathrm X}_{\mathrm{op}}\le 5/4$ and
\begin{equation}
    \mathsf S(\theta)=\mathrm X\mathrm X^\dagger\preceq2\mathsf I_{\mathfrak g_n}.
\end{equation}
This proves the pointwise bound \cref{eq:framebound}.

It remains to bound the number of rejections and the corresponding probability. 
We introduce the filtration $(\mathcal F_k)_{k=0}^T$, where $\mathcal F_k$ is the $\sigma$-algebra generated by the first $k$ Pauli operators and rotation angles. 
Recall that $X_j=\Ad_{U_{<j}^\dagger}(P_j)$ and that $P_j$ is chosen independently of $\mathcal{F}_{j-1}$; using \cref{eq:orthonormalP}, a straightforward calculation shows:
\begin{equation}\label{eq:rpr-boundedXX}
  \EX*[][\mathcal F_{j-1}]{\ketbra{X_j}}=\EX{\ketbra{P_j}} =  \frac1D \mathsf I_{\mathfrak g_n}.
\end{equation}
Also note that $\mathcal{W}_{j-1}$ is $\mathcal{F}_{j-1}$-measurable; we therefore get:
\begin{equation}
  \qquad
  \EX*[][\mathcal F_{j-1}]{
    \norm{\operatorname{Proj}_{\mathcal W_{j-1}}X_j(\theta)}_2^2
  }
  =\frac{\dim\mathcal W_{j-1}}D\le\frac T{D}.
\end{equation}
Markov's inequality then implies
\begin{equation}
  \prob{a_j=0\mid\mathcal F_{j-1}}   \le p_{\lowT}:= \frac{32T^2}{D}.
\end{equation}

Applying the chain rule, we know that $\prob{a_j=s_j~\forall j}\le p_{\lowT}^{T/2}$ for any fixed sequence $(s_j)_{j=1}^T\in\{0,1\}^T$ with at least $T/2$ zeros.
The union bound then implies:
\begin{equation}
  \prob*{\sum_{j=1}^T a_j<\frac{T}{2}}
  \le2^T p_{\lowT}^{T/2}\leq \exp(-\Omega(T)).
\end{equation}
Here in the last step, we use $T\le 2^{n-4}$.
\end{proof}

\subsection{Proof of \texorpdfstring{\cref{cor:rpr-anticoncentration}}{the RPR small-ball theorem}}
We can now combine the above two lemmas to obtain the small-ball estimate for random Pauli rotations.

\begin{proof}
\Cref{prop:selector} supplies the functions $a_j$ satisfying the hypotheses of \cref{lem:anticoncentration} with $B=2$.
First conditioning on the Pauli operators, applying \cref{lem:anticoncentration}, then averaging over the Pauli operators and applying \cref{eq:rprKbound}, we get:
\begin{equation}
\begin{aligned}
  \prob*{d_2(U(\theta),V)\le 1}
  &\le\prob*{\sum_{j=1}^T a_j< \frac{T}{2}}
    +e^{-T/32}
    \le \exp(-\Omega(T)).
\end{aligned}
\end{equation}
The bound is uniform in $V$.
\end{proof}

\section{All-to-all 2-local random unitary circuits}
\label{sec:all-to-all-local}

In this section, we consider the all-to-all 2-local random unitary circuits.
Let $T$ be the total number of gates and let $U_T$ denote the output of the whole random circuit.
The main result in this section is the following small-ball estimate.

\begin{theorem}\label{thm:local-setwise-anticoncentration}
There exists an absolute constant $c_0>0$ such that for all $1\le T\le D$ and every nonempty finite set $\mathcal V\subset \mathbf U(N)$ satisfying $\log\abs{\mathcal V}\le c_0T$, one has
\begin{equation}\label{eq:local-setwise-anticoncentration}
  \prob*{\min_{V\in\mathcal V}d_2(U_T,V)\le1}
  \le 
\exp\qty[-\Omega\qty(\frac{T}{n^2}(1+\log\frac{D}{T}))].
\end{equation}
Furthermore, if $T=O(n^2)$ and $\log\abs{\mathcal V}=O(T)$ with sufficiently small proportional constants, the r.h.s. can be improved to $\exp(-\Omega(T))$.
\end{theorem}

The proof follows the same overall strategy as that of \cref{sec:rpr}, although the technical details are more involved.
Let us first set up more notation.

We divide the circuits into blocks, each of which has length $\Theta(n^2)$.
Each block consists of a ``mixing'' segment of length $m$, followed by an ``active'' segment of length exactly $d$, where
\begin{equation}
  m=C_{\mathrm{mix}}n^2,\qquad d=\binom{n}{2}.
\end{equation}
The reason for the choice of $C_{\mathrm{mix}}$ and the terminology here will become clear later.
We denote the number of blocks by $R=T/(m+d)$.
(We may first assume $R\in\mathbb Z$ and $R\geq 1$ for convenience; general cases will be considered at the end.)
The total number of active gates is $\Tact=dR$; in particular, $\Tact=\Theta(T)$.

\begin{figure}[t]
  \centering
  \includegraphics[width=\linewidth]{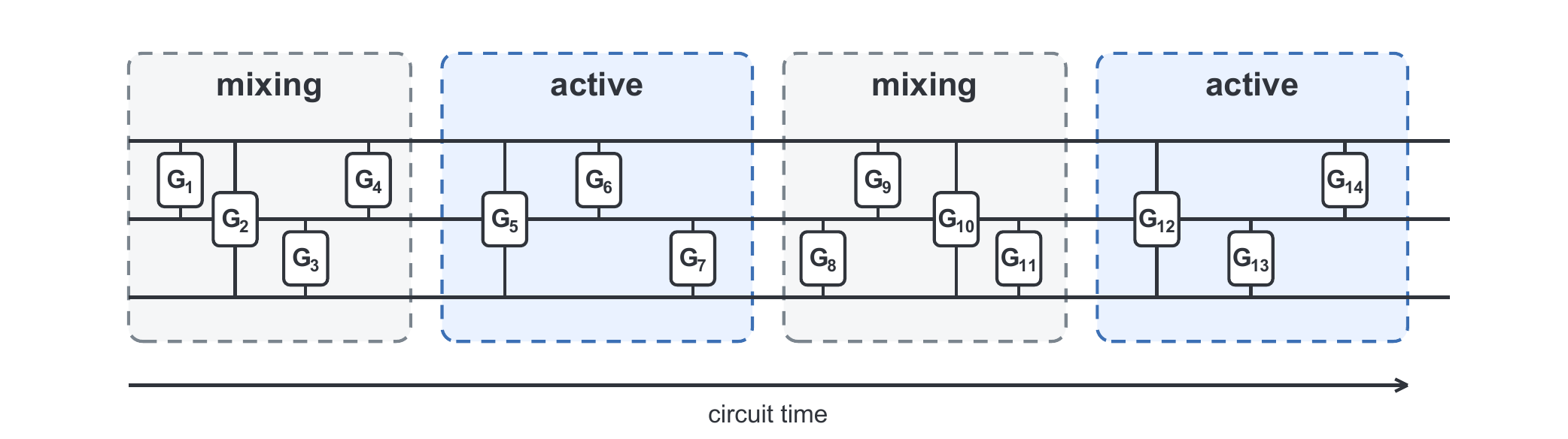}
  \caption{Schematic decomposition of the circuit into $R$ blocks, each consisting of a mixing segment of length $m$ followed by an active segment of length $d$. Thus, the total number of gates is $T=(m+d)R$.}
  \label{fig:local-blocks}
\end{figure}

The space of one gate and the circuit path space are
\begin{equation}
  \mathrm G_n
  =  \bigsqcup_{e\in\binom{[n]}2}\mathbf U(4)_e, %\{e\}\times
  \qquad
  \Omega_T=(\mathrm G_n)^T.
\end{equation}
% Conditional on the support sequence $(e_1,\ldots,e_t)$, the continuous fiber is a product of Haar spaces. 
For each pair $e$, we fix a weight-two Pauli operator $P_e$, such as $Z\otimes Z$ on that pair. 
For the $j$-th gate ($j\in[T]$), conditional on the support $e_j$, let $\partial_j$ be the left-invariant vector field generated by the one-parameter flow $G_j\longmapsto G_j e^{ihP_{e_j}}$.
These componentwise definitions combine to give a vector field on $\mathrm G_n$, and then a vector field on $\Omega_T$. They are also denoted by $\partial_j$.

Denote $G_{<j}=G_{j-1}\cdots G_1$; a direct calculation gives
\begin{equation}\label{eq:local-path-differential}
  \partial_jU_T=iU_TX_j,
  \qquad X_j=\Ad_{G_{<j}^\dagger}(P_{e_j}).
\end{equation}
Importantly, $X_j$ only depends on the previous gates $G_{<j}$ and the support of the current gate $G_j$, and is independent of $G_j$ conditional on the support.

We will select gates from the \emph{active} segments.
For each $j\in [T]$, let $a_j:\Omega_T\to\{0,1\}$ indicate whether the $j$-th gate is selected ($a_j=0$ automatically if $j$ is in  mixing segments).
Define
\begin{equation}\label{eq:local-load-and-count}
  \mathsf S=\sum_{j}a_j\ketbra{X_j}\in\operatorname{End}(\mathfrak g_n),
  \qquad
  K=\Tr \mathsf S = \sum_{j}a_j \in \mathbb N.
\end{equation}
Note that $a_j$, $\mathsf S$, and $K$ will depend adaptively on the gate sequence.
(One may view them as functions on $\Omega_T$.)
We will ensure that $a_j$ is independent of $G_j$, conditional on its support. In particular, $\partial_j a_j=0$.

\subsection{Small-ball estimates via bounded frame operators}

The next lemma isolates the analytic step that converts bounds on $\mathsf S$ and $K$ into a small-ball estimate.
It is a counterpart of \cref{lem:anticoncentration}.

\begin{lemma}\label{lem:local-stationary-small-ball}
Suppose that for any gate sequence,
\begin{equation}\label{eq:Sassumption}
  \mathsf S\preceq B\mathsf I_{\mathfrak g_n}.
\end{equation}
Also suppose that $\partial_j a_j=0$ ($\forall j\in[T]$).
Then for every $k>0$ and every nonempty finite set $\mathcal V\subset \mathbf U(N)$ satisfying 
\begin{equation}
  \log\abs{\mathcal V}\le k/(64B),
\end{equation}
it holds that
\begin{equation}\label{eq:local-stationary-small-ball}
  \prob*{\min_{V\in\mathcal V}d_2(U_T,V)\le1}
  \le 2\prob*{K<k}+\exp(-\frac{k}{64B}).
\end{equation}
\end{lemma}

\begin{proof}
For each $V\in\mathcal V$, set
\begin{equation}
  f^V=\Re\tau(V^\dagger U_T),
  \qquad
  W_j^V:=\partial_j f^V.
\end{equation}
Since $X_j^2=I$, we have
\begin{equation}\label{eq:local-harmonic-identity}
  \partial_j^2 f^V=-f^V.
\end{equation}
Moreover, the same calculation as in the previous section shows that the
pointwise bound $\mathsf S\preceq B\mathsf I_{\mathfrak g_n}$ implies
\begin{equation}\label{eq:local-gradient-load}
  \sum_{j}a_j(W_j^V)^2\le B
  \qquad\text{for every }V\in\mathcal V.
\end{equation}

Define the operator $\mathcal L$ and its carr\'e du champ \cite{BakryGentilLedoux2014} by
\begin{equation}
  \mathcal L=\sum_{j}a_j\partial_j^2,
  \qquad
  \Gamma(F)=\sum_{j}a_j(\partial_jF)^2.
\end{equation}
Since $a_j$ is independent of $G_j$ conditional on its support, integration by parts gives, for any functions $\phi:\mathbb{R}\to\mathbb{R}$ and $F:\Omega_T\to\mathbb{R}$, 
\begin{equation}\label{eq:local-stationary-identity}
  -\EX{\phi(F)\mathcal LF}
  =\EX{\phi'(F)\Gamma(F)}.
\end{equation}

Our goal is to understand the distribution tail of $\max_V f^V$.
We introduce a smooth maximum, the normalized log-sum-exp (see, e.g. \cite{chatterjee2005error}), where $\beta$ will be determined later:
\begin{equation}\label{eq:local-soft-maximum}
  f_\beta
  =\frac1\beta\log(\frac1{\abs{\mathcal V}}
    \sum_{V\in\mathcal V}e^{\beta f^V}).
\end{equation}
We have 
\begin{equation}\label{eq:softmax}
  f_\beta \le \max_V f^V\le f_\beta+\frac{\log \abs{\mathcal V}}{\beta}.
\end{equation}

We record two estimates to be used below.
First,
\begin{equation}\label{eq:local-soft-drift}
  -\mathcal Lf_\beta\ge Kf_\beta-\beta B.
\end{equation}
To prove it, we differentiate \cref{eq:local-soft-maximum}, and use \cref{eq:local-harmonic-identity}, yielding
\begin{align}\label{eq:local-soft-generator}
  \mathcal Lf_\beta=-K\overline f+\beta R_\beta
\end{align}
where
\begin{align}
  \overline f= \sum_{V\in\mathcal V}p_V f^V,
  \quad
    p_V=\frac{e^{\beta f^V}}{\sum_{W\in\mathcal V}e^{\beta f^W}},
\quad
  R_\beta
  =\sum_{j}a_j
    \qty(\sum_Vp_V(W_j^V)^2
      -\qty(\sum_Vp_VW_j^V)^2).
\end{align}
The fact that entropy is maximized at an even distribution gives $\overline f\ge f_\beta$. \Cref{eq:local-gradient-load} gives $R_\beta\le B$.
This proves the claimed bound \cref{eq:local-soft-drift}.
Second, we have
\begin{equation}\label{eq:local-bound-Gamma}
  \Gamma(f_\beta) \leq B.
\end{equation}
This is due to $(\partial_j f_\beta)^2 = (\sum_V p_V W_j^V)^2 \leq \sum_V p_V (W_j^V)^2$ and \cref{eq:local-gradient-load}.

Let
\begin{equation}
  \nu=\prob{K<k},
  \qquad H(u)=\prob{f_\beta \ge u}.
\end{equation}
Take $\phi(x)=\mathbf1_{\{x\ge u\}}$, so that $\phi'(x)=\delta(x-u)$.
It follows that, for $u>0$,
\begin{align}
  -BH'(u) = \EX{\delta(f_\beta-u) B}
  &\ge \EX{\phi'(f_\beta)\Gamma(f_\beta)}
   =-\EX{\mathbf1_{\{f_\beta>u\}}\mathcal Lf_\beta}\notag\\
  &\ge ku\qty(H(u)-\nu)-\beta BH(u).
  \label{eq:local-tail-differential}
\end{align}
Here we use \cref{eq:local-bound-Gamma} in the first inequality, and \cref{eq:local-soft-drift} in the second inequality.
Set $\beta=k/(8B)$; then, for $u\ge 1/4$,
\begin{equation}
  -BH'(u)\ge\frac{ku}{2}\qty(H(u)-2\nu).
\end{equation}
Integrating from $1/4$ to $3/8$ gives
\begin{equation}\label{eq:local-soft-tail}
  \prob*{f_\beta\ge\frac38}
  \le 2\nu+\exp(-\frac{5k}{256B})
  \le 2\nu+\exp(-\frac{k}{64B}).
\end{equation}

Since
\begin{equation}
  d_2(U_T,V)^2=2-2f^V,
\end{equation}
$d_2(U_T,V)\le1$ implies $f^V\ge1/2$. Hence, using the bound \cref{eq:softmax}, we have
\begin{equation}
  f_\beta\ge\frac12-\frac{\log\abs{\mathcal V}}{\beta}\ge \frac38.
\end{equation}
Here we used the assumption $\log\abs{\mathcal V}\le k/(64B)=\beta/8$.
Therefore \cref{eq:local-soft-tail} gives
\begin{equation}
  \prob*{\min_{V\in\mathcal V}d_2(U_T,V)\le1}\le \prob*{f_\beta\ge\frac38}
    \le 2\nu+\exp\qty(-\frac{k}{64B}),
\end{equation}
which is the claimed estimate.
\end{proof}

\subsection{Selecting gates}

Next, we describe an algorithm to select a set of vectors satisfying the conditions required for \cref{lem:local-stationary-small-ball}.
There are two tests at each active gate: the first decides whether a vector should be accepted ``on average'', and the second decides whether to accept the actual vector.

Let $\mathcal F_j$ be the forward filtration generated by the first $j$ gate variables.
Recall that $D=\dim\mathfrak g_n$ and $d=\binom{n}{2}$.
We define
\begin{equation}\label{eq:local-segment-covariance}
  \mathsf C_0=\frac1d\sum_e\ketbra{P_e},
  \qquad
  \mathsf C_j=\EX*[][\mathcal F_{j-1}]{\ketbra{X_j}}
  =\Ad_{G_{<j}^\dagger}\mathsf C_0\Ad_{G_{<j}}.
\end{equation}
We have $\Tr\mathsf C_j = \Tr\mathsf C_0 =1$ and $\mathsf C_j\preceq \mathsf I/d$.
The operator $\mathsf C_j$ is determined by $\mathcal F_{j-1}$.

We will keep two running sums $\mathsf H$ and $\mathsf S$. 
The $\mathsf H$ will sum up $\mathsf C_j$ for those $j$s that pass the first test;
$\mathsf S$ will sum up the actual $\ketbra{X_j}$ for those $j$s that also pass the second test.
We fix a sufficiently small absolute constant $\alpha>0$ (to be determined later) and define
\begin{equation}\label{eq:local-covariance-potential}
  \Phi(\mathsf H)=\frac1\lambda\log\frac{\Tr e^{\lambda\mathsf H}}D,
  \qquad \text{where~}
  \lambda=\alpha n.
\end{equation}
It may be viewed as a soft version of $\norm{\mathsf H}_{\mathrm{op}}$.

We initialize $\mathsf H=\mathsf S=\mathsf 0$.
Then we repeat the following steps sequentially for all active gates. 
At the beginning of each active segment $r$, store the current value of
$\mathsf H$ as $\mathsf H_r^{\mathrm{start}}$.  
For each gate $j$ in segment $r$, perform the following steps:
\begin{enumerate}
  \item \textbf{First test.} Set
  \begin{equation}\label{eq:local-covariance-selector}
    b_j=\begin{cases}
      1, & \text{if } \Phi(\mathsf H+\mathsf C_j)
        \le \Phi(\mathsf H_r^{\mathrm{start}})+128/R,\\
      0, & \text{otherwise},
    \end{cases}
  \end{equation}
  and update
  \begin{equation}\label{eq:update1}
    \mathsf H\longleftarrow\mathsf H+b_j\mathsf C_j.
  \end{equation}
  If $b_j=0$, we set $a_j=0$, skip step 2, and proceed to the next active gate.

  \item \textbf{Second test.} If $b_j=1$, we compute
  \begin{equation}
    q_j = \expval{(16(\mathsf I+\mathsf H)-\mathsf S)^{-1}}{X_j},
  \end{equation}
  where $\mathsf H$ denotes its value after the preceding update \cref{eq:update1}.
  We then set
  \begin{equation}
    a_j=\begin{cases}
      1, & \text{if } q_j\le 1/2,\\
      0, & \text{otherwise},
    \end{cases}
  \end{equation}
  and update
  \begin{equation}
    \mathsf S\longleftarrow\mathsf S+a_j\ketbra{X_j}.
  \end{equation}
  We remark that $q_j\le 1/2$ if and only if $\ketbra{X_j}\preceq \tfrac12 (16(\mathsf I+\mathsf H)-\mathsf S)$.
\end{enumerate}

Note that $b_j$ is fully determined by previous gates, and $a_j$ is independent of its own Haar gate conditional on its support $e_j$, as required.

\subsubsection{Analysis of the first selector}

We use the following three lemmas to establish that with large probability, at least a constant portion of the active gates pass the first test (\cref{lem:local-covariance-count}).

\begin{lemma}\label{lem:local-isotropization}
There is an absolute constant $C_{\mathrm{mix}}$ such that, for
$m=C_{\mathrm{mix}}n^2$ and an all-to-all 2-local random unitary $U_m$
consisting of $m$ gates,
\begin{equation}\label{eq:local-isotropization}
  \EX{\Ad_{U_m} \mathsf A \Ad_{U_m}^{-1}}
  \preceq \frac{2\Tr \mathsf A}{D}\mathsf I_{\mathfrak g_n},
\end{equation} 
where $\mathsf A$ is any deterministic positive semidefinite operator on $\mathfrak g_n$.
\end{lemma}

\begin{proof}
  We define two (super)operators on $\operatorname{End}(\mathfrak g_n)$:
\begin{equation}\label{eq:local-second-moment}
  \mathcal M(\mathsf A)=\EX{\Ad_{G^\dagger}\mathsf A\Ad_G},
  \qquad
  \mathcal P(\mathsf A)=\frac{\Tr\mathsf A}{D}\mathsf I,
\end{equation}
where $G$ is one 2-local Haar gate and $\mathcal P$ is the orthogonal projection onto the scalar operators.
$\mathcal M$ is positive semidefinite and self-adjoint, since Haar averaging over each fixed edge gives an orthogonal projection.
Furthermore, the 2-design spectral gap \cite{HarrowLow2009,DinizJonathan2011} gives, for some $\gamma=\Omega(1/n)$,
\begin{equation}\label{eq:local-transfer-gap}
  0\preceq\mathcal M\preceq(1-\gamma)\mathcal I+\gamma\mathcal P.
\end{equation}
Taking $m$ steps contracts the orthogonal complement of the scalar operators further:
\begin{equation}\label{eq:2.8}
  \norm{
    \EX{\Ad_{U_m} \mathsf A \Ad_{U_m}^{-1}}
    -\frac{\Tr \mathsf A }D\mathsf I_{\mathfrak g_n}
  }_{\mathrm F}
  \le e^{-\gamma m}
  \norm{
    \mathsf A -\frac{\Tr \mathsf A }D\mathsf I_{\mathfrak g_n}
  }_{\mathrm F}.
\end{equation}
Here we write $\norm{\cdot}_{\mathrm F}$ for the (unnormalized) Frobenius norm of linear operators on $\mathfrak g_n$:
\begin{equation}
  \norm{\mathsf A}_{\mathrm F} = (\operatorname{Tr}(\mathsf A^\dagger \mathsf A))^{1/2}.
\end{equation}

Since $\mathsf A \succeq0$, we can bound the r.h.s. as
\begin{equation}
  e^{-\gamma m}
  \norm{\mathsf A -\frac{\Tr \mathsf A}D\mathsf I}_{\mathrm F}
  \le e^{-\gamma m}\sqrt{\Tr(\mathsf A^2)}
  \le e^{-c C_{\mathrm{mix}}n}\Tr \mathsf A.
\end{equation}
Choose $C_{\mathrm{mix}}$ large enough so that
$e^{-c C_{\mathrm{mix}}n}\le D^{-1}$.
The claimed operator inequality \cref{eq:local-isotropization} follows from
$\norm{\cdot}_{\mathrm{op}}\le\norm{\cdot}_{\mathrm F}$.
\end{proof}

\begin{lemma}\label{lem:local-covariance-moment}
  Fix an active segment $r$, and let $\mathsf H_r^{\mathrm{raw}}$ be the sum of all $\mathsf C_j$ for $j$ in segment $r$.
For $\lambda=\alpha n$ with $\alpha>0$ sufficiently small, we have
\begin{equation}\label{eq:local-block-exponential-increment}
  \EX*[][\text{past blocks}]{e^{\lambda\mathsf H_r^{\mathrm{raw}}}-\mathsf I}
  \preceq\frac{16\lambda d}{D}\mathsf I.
\end{equation}
\end{lemma}
Note that here the ``past blocks'' in the condition do not include that particular mixing segment immediately preceding the active segment $r$.
% This lemma may be viewed as a counterpart of \cref{eq:rpr-boundedXX}.

\begin{proof}
First, for convenience, let us remove all gates before the active segment $r$. 
We number the gates in this segment as $G_1,\ldots,G_d$, so that
\begin{equation}
    \mathsf H_r^{\mathrm{raw}}=\sum_{j=1}^d\mathsf C_j,
  \qquad
  \mathsf C_j=\Ad_{G_{<j}^\dagger}\mathsf C_0\Ad_{G_{<j}}.
\end{equation}
We use the operators $\mathcal M,\mathcal P$ defined in \cref{eq:local-second-moment} and the spectral-gap bound \cref{eq:local-transfer-gap}.

The multivariate Golden--Thompson inequality \cite{SutterBertaTomamichel2017} gives
\begin{equation}
  \Tr e^{\lambda\mathsf H_r^{\mathrm{raw}}}
  \le\int_{\mathbb R}\beta_0(t)
  \Bigl\|\prod_{j=1}^d e^{(1+it)\lambda\mathsf C_j/2}\Bigr\|_{\mathrm F}^2\,dt,
  \qquad\text{where~}
  \mathsf B_t=e^{(1+it)\lambda\mathsf C_0/2}.
\end{equation}
Here $\beta_0$ is a probability density.
The product inside the norm can be rewritten as:
\begin{equation}
  \begin{aligned}
      &\mathsf B_t \cdot \Ad_{G_1^\dagger}\mathsf B_t \Ad_{G_1} \cdot \Ad_{G_1^\dagger G_2^\dagger}\mathsf B_t \Ad_{G_2 G_1} \cdots \Ad_{G_{<d}^\dagger} \mathsf B_t\Ad_{G_{<d}}\\
  =&
  \mathsf B_t\Ad_{G_1^\dagger}\mathsf B_t\cdots\Ad_{G_{d-1}^\dagger}\mathsf B_t\Ad_{G_{<d}}
  \end{aligned}
\end{equation}
Taking the average over the independent gates $G_1$ to $G_{d-1}$, we get:
\begin{equation}
  \EX*{\norm{\mathsf B_t\Ad_{G_1^\dagger}\mathsf B_t\cdots\Ad_{G_{d-1}^\dagger}\mathsf B_t}_{\mathrm F}^2}
  =\Tr\mathcal T_t^{d-1}(e^{\lambda\mathsf C_0}),
  \qquad\text{where~}
  \mathcal T_t(\cdot)=\mathsf B_t\mathcal M(\cdot)\mathsf B_t^*.
\end{equation}

To bound this expression, we use the decomposition
\begin{equation}
  \mathcal E^{-1}\mathcal T_t\mathcal E=\mathcal U_t\mathcal K
\end{equation}
where
\begin{equation}
  \mathcal E(\cdot)=e^{\lambda\mathsf C_0/4}(\cdot) e^{\lambda\mathsf C_0/4},
  \qquad
    \mathcal K=\mathcal E\mathcal M\mathcal E,
  \qquad
  \mathcal U_t(\cdot)=e^{it\lambda\mathsf C_0/2}(\cdot) e^{-it\lambda\mathsf C_0/2}.
\end{equation}
Note that $\mathcal U_t$ is unitary; using Cauchy--Schwarz, we have:
\begin{equation}
  \begin{aligned} 
  \Tr\mathcal T_t^{d-1}\bigl(e^{\lambda\mathsf C_0}\bigr)
  &= \Tr\!\left[\mathcal E(\mathcal U_t\mathcal K)^{d-1}
    \mathcal E^{-1}\bigl(e^{\lambda\mathsf C_0}\bigr)\right] 
   = \Tr\!\left[e^{\lambda\mathsf C_0/2} (\mathcal U_t\mathcal K)^{d-1}(
    e^{\lambda\mathsf C_0/2})\right]\\
  &\le \norm{\mathcal K}^{d-1}_{\mathrm{op}}\Tr e^{\lambda\mathsf C_0}.
\end{aligned}
\end{equation}
We therefore obtain
\begin{equation}\label{eq:local-transfer-reduction}
  \EX*{\Tr e^{\lambda\mathsf H_r^{\mathrm{raw}}}}
  \le \int_{\mathbb R}\beta_0(t)
    \Tr\mathcal T_t^{d-1}(e^{\lambda\mathsf C_0})\dd t
  \le \norm{\mathcal K}^{d-1}_{\mathrm{op}} \Tr e^{\lambda\mathsf C_0}.
\end{equation}

Now we estimate $\norm{\mathcal K}_{\mathrm{op}}$.
Take $\alpha$ small enough so that $\lambda/d\le\gamma/4$.
By \cref{eq:local-transfer-gap},
\begin{equation}\label{eq:opkbound}
  0\preceq\mathcal K
  \preceq(1-\gamma)\mathcal E^2+\gamma\mathcal E\mathcal P\mathcal E.
\end{equation}
On the orthogonal complement of $\operatorname{span}\{\mathsf I-\mathsf\Pi,\mathsf\Pi\}$ (here $\mathsf\Pi=d\mathsf C_0$ is an orthogonal projector), a straightforward calculation shows that the image of $\mathcal E$ is always traceless, hence $\mathcal P \mathcal E=0$; the operator norm of $\mathcal E$ is upper bounded by $e^{\lambda/(2d)}$.
Therefore, the r.h.s. of \cref{eq:opkbound} has operator norm at most $(1-\gamma)e^{\lambda/d}<1$.
On $\operatorname{span}\{\mathsf I-\mathsf\Pi,\mathsf\Pi\}$, in the normalized basis, its matrix is
\begin{equation}
  \begin{pmatrix}
    1-\gamma\eta & \gamma\sqrt{\eta(1-\eta)(1+v)}\\
    \gamma\sqrt{\eta(1-\eta)(1+v)} & (1-\gamma+\gamma\eta)(1+v)
  \end{pmatrix}
  \preceq(1+2\eta v)I_2.
\end{equation}
Here $\eta=d/D$ and $v=e^{\lambda/d}-1\le2\lambda/d\le\gamma/2$, and the last inequality can be checked directly.
Consequently, 
\begin{equation}
  \norm{\mathcal K}_{\mathrm{op}}\le1+2\eta v.
\end{equation}
Using \cref{eq:local-transfer-reduction}, we obtain
\begin{equation}
  \EX*{\Tr e^{\lambda\mathsf H_r^{\mathrm{raw}}}} 
\le (1+2\eta v)^{d-1} \Tr e^{\lambda\mathsf C_0}
  = D(1+\eta v)(1+2\eta v)^{d-1}
  \le D\exp(\frac{4\lambda d}{D}).
\end{equation}
Subtracting $D$ and noting that $\lambda/d\le\gamma/4$ implies $\lambda d/D\le 1/4$, we get 
\begin{equation}\label{eq:local-covariance-excess}
  \EX*{\Tr(e^{\lambda\mathsf H_r^{\mathrm{raw}}}-\mathsf I)}\le8\lambda d.
\end{equation}

Now back to the actual circuit. 
Write $U_1$ for the circuit formed by the earlier blocks, $U_2$ for the preceding mixing circuit, and $\widetilde{\mathsf H}_r$ for the covariance sum estimated above (with all previous segments deleted).
The actual covariance sum for segment $r$ is
\begin{equation}
  \mathsf H_r^{\mathrm{raw}}
  =\Ad_{U_1^\dagger}\Ad_{U_2^\dagger}
    \widetilde{\mathsf H}_r\Ad_{U_2}\Ad_{U_1}.
\end{equation}
For each fixed realization of the active gates, \cref{lem:local-isotropization} gives
\begin{equation}
  \EX*[U_2]{\Ad_{U_2^\dagger}
    (e^{\lambda\widetilde{\mathsf H}_r}-\mathsf I)\Ad_{U_2}}
  \preceq\frac{2}{D}
    \Tr(e^{\lambda\widetilde{\mathsf H}_r}-\mathsf I)\mathsf I,
\end{equation}
where the expectation is over the mixing gates only.
Conjugating by $\Ad_{U_1^\dagger}$ and averaging over the active gates, \cref{eq:local-covariance-excess} therefore yields
\begin{equation}
  \EX*[][\text{past blocks}]{e^{\lambda\mathsf H_r^{\mathrm{raw}}}-\mathsf I}
  \preceq\frac{2}{D}
    \EX*{\Tr(e^{\lambda\widetilde{\mathsf H}_r}-\mathsf I)}\mathsf I
  \preceq\frac{16\lambda d}{D}\mathsf I.
\end{equation}
\end{proof}

\begin{lemma}\label{lem:local-covariance-count}
% Assume $T\le D$.
The number $N_b:=\sum_jb_j$ of gates passing the first test satisfies
\begin{equation}\label{eq:local-covariance-count}
  \prob*{N_b<\frac{\Tact}{2}}\le\qty(\frac{\Tact}{2D})^{R/2}.
\end{equation}
\end{lemma}

\begin{proof}
Let
\begin{equation}
  \Delta_r=\Phi(\mathsf H_r^{\mathrm{start}}+\mathsf H_r^{\mathrm{raw}})
    -\Phi(\mathsf H_r^{\mathrm{start}}).
\end{equation}
By Golden--Thompson inequality \cite{Golden1965,Thompson1965},
\begin{align}
  e^{\lambda\Delta_r}
  =\frac{\Tr e^{\lambda(\mathsf H_r^{\mathrm{start}}
      +\mathsf H_r^{\mathrm{raw}})}}
          {\Tr e^{\lambda\mathsf H_r^{\mathrm{start}}}} 
  \le
    \frac{\Tr\qty(e^{\lambda\mathsf H_r^{\mathrm{start}}}
      e^{\lambda\mathsf H_r^{\mathrm{raw}}})}
         {\Tr e^{\lambda\mathsf H_r^{\mathrm{start}}}}.
\end{align}
Since $\mathsf H_r^{\mathrm{start}}$ is fixed under the conditional expectation,
\cref{eq:local-block-exponential-increment} therefore gives
\begin{align}
  \EX*[][\text{past blocks}]{e^{\lambda\Delta_r}-1}
  \le \frac{16\lambda d}{D}.
\end{align}

If $\Delta_r\le128/R$, every partial sum will satisfy the constraint in \cref{eq:local-covariance-selector}, so no gate in the segment will be rejected.
Markov's inequality therefore gives
\begin{equation}
  \prob*{\text{some $b_j=0$ in segment $r$}\mid\text{past blocks}}
  \le\frac{16\lambda d/D}{e^{128\lambda/R}-1}
  \le\frac{\Tact}{8D}.
\end{equation}
If $N_b<\Tact/2$, at least half of the $R$ segments must contain a rejection.
The union bound and chain rule therefore give
\begin{equation}
  \prob*{N_b<\frac{\Tact}{2}}
  \le2^R\qty(\frac{\Tact}{8D})^{R/2}
  =\qty(\frac{\Tact}{2D})^{R/2}.
\end{equation}
\end{proof}

\subsubsection{Analysis of the second selector}

After applying both tests to all gates, we write $\mathsf H_T$ and $\mathsf S_T$ for the final values of the running sums.
Denote $K=\Tr\mathsf S_T=\sum_ja_j$, which is the total number of gates that pass both tests.

By definition of the first selector, each active segment increases $\Phi$ by at most $128/R$, so the final sum after checking all gates satisfies $\Phi(\mathsf H_T)\leq128$ (note that $\Phi(\mathsf H)=0$ initially). 
This implies:
\begin{equation}\label{eq:local-covariance-load}
  \mathsf H_T\preceq \tilde B\mathsf I,
  \qquad\text{where~}
  \tilde B=128+\frac{\log D}{\lambda}=O(1).
\end{equation}

The next lemma shows that the second selector also preserves a uniform bound of $\mathsf S$, and rejects only a small number of gates with high probability.
\begin{lemma}\label{lem:local-global-selection}
After the selection algorithm, we have
\begin{equation}\label{eq:local-pointwise-load}
  \mathsf S_T\preceq B\mathsf I_{\mathfrak g_n},
\end{equation}
where $B=16(1+\tilde B)=O(1)$.
Moreover,
\begin{equation}\label{eq:local-common-bad-probability}
  \prob*{K<\frac{\Tact}{4}}
  \le\qty(\frac{\Tact}{2D})^{R/2}+e^{-\Tact/16}.
\end{equation}
\end{lemma}

\begin{proof}
Write $\mathsf M=16(\mathsf I+\mathsf H)-\mathsf S$.
We claim that $\mathsf M\succ \mathsf 0$ throughout the construction.
Initially $\mathsf M=16\mathsf I$.
If a gate $j$ does not pass the first test, then $\mathsf M$ remains invariant.
If a gate $j$ passes the first test but not the second, then $\mathsf M\longleftarrow \mathsf M + 16\mathsf C_j$, and is still positive.
If a gate $j$ passes both tests, then by definition,
\begin{equation}
  \mathsf M_j + 16\mathsf C_j \succ \ketbra{X_j}
\end{equation}
(here $\mathsf M_j$ is evaluated before the first update), 
hence $\mathsf M$ is still positive after two updates.
Therefore, using \cref{eq:local-covariance-load} and $\mathsf{M}\succ \mathsf 0$, $\mathsf S \preceq 16(\mathsf I+\mathsf H) \preceq 16(1+\tilde B)\mathsf I$, which is \cref{eq:local-pointwise-load}.

It remains to bound the number of rejections in the second test.
For $b_j=1$, the first test update increases $\log\det\mathsf M$ by
\begin{equation}
  \delta_j=\log\det(\mathsf M+16\mathsf C_j)-\log\det\mathsf M,
\end{equation}
where $\mathsf M$ is evaluated before the update. 
The conditional covariance identity gives
\begin{align}
  \EX*[][\mathcal F_{j-1}]{q_j}
  &=\Tr\qty[\mathsf C_j(\mathsf M+16\mathsf C_j)^{-1}]\\
  &\le\frac1{16}\int_0^{16}\Tr\qty[\mathsf C_j(\mathsf M+t\mathsf C_j)^{-1}]\,dt
  =\frac{\delta_j}{16}.
\end{align}
Let $p_j=\EX*[][\mathcal F_{j-1}]{b_j-a_j}$.
Recall that $b_j$ is determined by $\mathcal F_{j-1}$.
If $b_j=1$, Markov's inequality gives $p_j=\prob*{q_j>1/2\mid\mathcal F_{j-1}}\le \delta_j/8$; if $b_j=0$ then $p_j=0$ as well.

Let us bound $\sum_j p_j$ by bounding the accumulated log-determinant increments.
Write $\mathsf M_j^+=\mathsf M+16\mathsf C_j$.
Before and after the second test, an accepted direction ($q_j\le1/2$) can decrease the log-determinant by at most $\log2$, since
\begin{equation}
  \log\det(\mathsf M_j^+-\ketbra{X_j}) = \log \det(\mathsf M_j^+) + \log(1-q_j)
  \ge \log \det(\mathsf M_j^+) - \log 2.
\end{equation}
Summing over all updates, with $\mathsf M_T$ denoting its final value, yields
\begin{equation}
  \log\det\mathsf M_T - \log\det 16\mathsf I
  \geq
  \sum_{j:b_j=1}\delta_j - K\log 2.
\end{equation}
Using $\mathsf M\preceq16(\mathsf I+\mathsf H)$ at the end, we obtain
\begin{equation}
  \sum_{j:b_j=1}\delta_j
  \le\log\det(\mathsf I+\mathsf H)+K\log2
  \le N_b+K\log2.
\end{equation}
Here we used $\log\det(\mathsf I+\mathsf H)\le\Tr\mathsf H=N_b$.
Consequently, on every path,
\begin{equation}\label{eq:local-rejection-compensator}
  \sum_jp_j\le\frac18(N_b+K\log2).
\end{equation}

The total number of rejections by the second selector is $N_b-K=\sum_j(b_j-a_j)$.
Since each summand is Bernoulli with conditional mean $p_j$,
\begin{equation}
  \EX*[][\mathcal F_{j-1}]{e^{-p_j}2^{b_j-a_j}}
  =e^{-p_j}(1+p_j)\le1.
\end{equation}
Iterating conditional expectations therefore gives
\begin{equation}
  \EX*{\exp\big( (\log2)(N_b-K)-\sum_jp_j\big)}
  =
  \EX*{\prod_j e^{-p_j}2^{b_j-a_j}}\le1.
\end{equation}
Note that, on the event $\{K<\Tact/4,\ N_b\ge \Tact/2\}$, \cref{eq:local-rejection-compensator} gives
\begin{equation}
  (\log2)(N_b-K)-\sum_jp_j
  \ge\qty(\log2-\frac18)N_b-\frac98(\log2)K
  \ge\frac{\Tact}{16}.
\end{equation}
Therefore, Markov's inequality now gives
\begin{equation}
  \prob*{K<\Tact/4,\ N_b\ge \Tact/2}\le \exp(-\Tact/16).
\end{equation}
Combining this with \cref{eq:local-covariance-count} by the union bound proves \cref{eq:local-common-bad-probability}.
\end{proof}

\subsection{Proof of \texorpdfstring{\cref{thm:local-setwise-anticoncentration}}{the local setwise small-ball theorem}}

\begin{proof}

We first assume $m+d$ divides $T$. 
By \cref{lem:local-global-selection}, the two-step selection gives $\mathsf S_T\preceq B\mathsf I$ with $B=O(1)$, and each $a_j$ is independent of $G_j$ conditional on its support.
The event $\{K<\Tact/4\}$ is determined entirely by the circuit path and the two selectors; it does not depend on the target set $\mathcal V$.

We apply \cref{lem:local-stationary-small-ball} with $k=\Tact/4$, choosing $c_0$ small enough that $\log\abs{\mathcal V}\le c_0T\le k/(64B)$, and use \cref{lem:local-global-selection}.
Since $\Tact=\Theta(T)$ and $R=\Theta(T/n^2)$, we obtain
\begin{equation}
  \prob*{\min_{V\in\mathcal V}d_2(U_T,V)\le1}
  \le 
  % \Cexp\qty[-c\frac{T}{n^2}\qty(1+\log\frac{D}{T})],
\exp\qty[-\Omega\qty(\frac{T}{n^2}(1+\log\frac{D}{T}))],
\end{equation}
after absorbing the $e^{-\Omega(T)}$ term and adjusting the constants.

If $T>m+d$ but $m+d$ does not divide $T$, we apply the above argument to the first $\lfloor T/(m+d)\rfloor$ blocks and adjust the constants.

Finally, suppose $T=O(n^2)$.
Set $\ell=\min\{T,d\}$. 
We regard the first $\ell$ gates as active and the rest as ``junk''.
We ignore the first test (simply set $b_j=1$ for all $j\leq \ell$) and apply the second test to the active gates, setting $a_j=0$ for $j>\ell$.

Since $\mathsf H=\sum_{j=1}^{\ell}\mathsf C_j\preceq(\ell/d)\mathsf I\preceq\mathsf I$, the same positive-definiteness argument for $\mathsf M$ gives
\begin{equation}
  \mathsf S\preceq32\mathsf I.
\end{equation}
The log-determinant and the exponential estimate in the proof of \cref{lem:local-global-selection}, with $N_b=\ell$ and $K<\ell/2$, now give
\begin{equation}
  \prob{K<\ell/2}\le e^{-\ell/8}.
\end{equation}
Applying \cref{lem:local-stationary-small-ball} with $k=\ell/2$ and $B=32$ yields, provided $\log\abs{\mathcal V}\le\ell/2^{12}$,
\begin{equation}
  \prob*{\min_{V\in\mathcal V}d_2(U_T,V)\le1}
  \le 2e^{-\ell/8}+e^{-\ell/2^{12}}.
\end{equation}
Since $\ell=\Theta(T)$, the desired result holds after adjusting the constants.
\end{proof}

\section{Complexity growth}\label{sec:complexity}

\begin{theorem}
  Take $\epsilon=1/2$. There exists an absolute constant $c>0$ such that for all $2\leq T\leq D$
  \begin{equation}
    \prob*{\mathcal C_\epsilon(U_T)\le c\frac{T}{\log T}}
\le \exp\qty[-\Omega\qty(\frac{T}{n^2}(1+\log\frac{D}{T}))].
 \end{equation}
  Furthermore, for $T=O(n^2)$, the r.h.s. can be improved to $\exp(-\Omega(T))$.
\end{theorem}

The proof is obtained by combining the two propositions below.
For $T\ge c_0n$, where $c_0$ is the absolute constant in \cref{prop:local-short-complexity}, we use \cref{cor:local-complexity-growth}, which follows from the small-ball estimate.
For $2\le T\le c_0n$, \cref{prop:local-short-complexity} gives a linear lower bound with failure probability $Ce^{-cT}$, which implies the claimed estimate after adjusting the constants.

\subsection{Large time}

The setwise small-ball estimate in \cref{sec:all-to-all-local} gives the following complexity lower bound.

\begin{proposition}\label{cor:local-complexity-growth}
Fix $c_0>0$ and take $\epsilon=1/2$. There exist constants $c,C>0$, depending only on $c_0$, such that for all $\max\{2,c_0n\}\leq T\leq D$
\begin{equation}\label{eq:local-complexity-growth}
  \prob*{\mathcal C_\epsilon(U_T)\le c\frac{T}{\log T}}
\le \exp\qty[-\Omega\qty(\frac{T}{n^2}(1+\log\frac{D}{T}))].
\end{equation}
Furthermore, for $T=O(n^2)$, the r.h.s. can be improved to $\exp(-\Omega(T))$.
\end{proposition}

\begin{proof}
Let $\mathcal U_r\subset\mathbf U(N)$ be the set of unitaries implementable by at most $r$ two-qubit gates ($r\geq 1$).
For $r\ge1$, choose an operator-norm $\epsilon/r$-net of $\mathbf U(4)$ with at most $(cr)^{16}$ elements, where $c$ is an absolute constant.
Replacing each gate by a net point changes the product by at most $\epsilon$ in operator norm, by a telescoping sum.
This gives an operator-norm $\epsilon$-net $\mathcal V_r$ of $\mathcal U_r$ satisfying
\begin{equation}\label{eq:local-circuit-net}
  \log\abs{\mathcal V_r}
  \le r\log\qty(\binom{n}{2}(cr)^{16})
  \le Cr\bigl(\log n+\log r\bigr).
\end{equation}
For $r=0$, take $\mathcal V_0=\{I\}$.

If $\mathcal C_\epsilon(U_T)\le r$, the triangle inequality and $d_2(U,V)\le\norm{U-V}_{\mathrm{op}}$ imply
\begin{equation}
  \min_{V\in\mathcal V_r}d_2(U_T,V)\le 2\epsilon = 1.
\end{equation}
Applying \cref{thm:local-setwise-anticoncentration} to $\mathcal V_r$ and using \cref{eq:local-circuit-net}, we obtain absolute constants $a,b,C>0$ such that, whenever $Cr(\log n+\log(1+r))\le bT$,
\begin{equation}\label{eq:local-complexity-tail}
  \prob*{\mathcal C_\epsilon(U_T)\le r}
  \le C\exp\qty[-a\frac{T}{n^2}(1+\log\frac{D}{T})].
\end{equation}
Since $T\ge c_0n$ and $T\ge2$, we have $\log(nT)\le C_{c_0}\log T$.
Choose $r=\lfloor cT/\log T\rfloor$, where $c>0$ is sufficiently small depending only on $c_0$, so that $Cr(\log n+\log(1+r))\le bT$.
This proves \cref{eq:local-complexity-growth}.

For $T=O(n^2)$, the improved estimate in \cref{thm:local-setwise-anticoncentration}, together with the same choice of $r$ and a smaller $c$ if necessary, improves the right-hand side to $\exp(-\Omega(T))$.
\end{proof}

\subsection{Short time}

For sufficiently short circuits, support collisions are rare, so many qubits are touched exactly once. Counting these qubits gives a linear complexity lower bound.

\begin{proposition}\label{prop:local-short-complexity}
There exists an absolute constant $c_0>0$ such that, for $2\le T\le c_0n$,
\begin{equation}
  \prob*{\mathcal C_{1/2}(U_T)\le cT}\le \exp(-\Omega(T)).
\end{equation}
\end{proposition}

\begin{proof}
Call the $j$-th gate a collision gate if its support $e_j$ intersects $e_1\cup\cdots\cup e_{j-1}$, and let $N_{\mathrm{col}}$ be the number of collision gates. Let $N_1$ be the number of qubits that belong to exactly one of the supports $e_1,\ldots,e_T$.
  A collision gate can reduce this count by at most 4, so
\begin{equation}
  N_1\ge 2T-4N_{\mathrm{col}}.
\end{equation}
Conditional on the first $j-1$ supports, the probability that $e_j$ intersects a previously used qubit is at most $4(j-1)/n\le4c_0$. The union bound and the chain rule give
\begin{equation}
  \prob{N_1 < T} \leq \prob{N_{\mathrm{col}}>T/4}
  \le2^T(4c_0)^{T/4}
  \le \exp(-\Omega(T))
\end{equation}
for sufficiently small $c_0$. 

Suppose $N_1\geq T$. We can then choose a set $\mathcal Q$ of $T/2$ qubits whose unique incident gates are all distinct (but such gates may collide at the other sides). 
For $a\in\mathcal Q$, denote its unique incident gate by $G_{j(a)}$. No other gates have support on $a$. Using unitary invariance of the operator norm, we have:
\begin{equation}
  \norm{[U_T,Z_a]}_{\mathrm{op}}
  =\norm{[G_{j(a)},Z_a]}_{\mathrm{op}},
\end{equation}
where $Z_a$ is the Pauli $Z$ operator on qubit $a$.
The gates $G_{j(a)}$ are Haar random, so $\norm{[G_{j(a)},Z_a]}_{\mathrm{op}}>1$ with an absolute probability $p>0$.
The Chernoff bound then shows that
\begin{equation}
  \prob*{\text{$\norm{[U_T,Z_a]}_{\mathrm{op}}>1$ for at least $\frac{pT}4$ qubits $a\in\mathcal Q$}} \geq 1-\exp(-\Omega(T)).
\end{equation}

If a circuit $W$ does not touch such a qubit $a$, then $[W,Z_a]=0$, so
\begin{equation}
  1<\norm{[U_T,Z_a]}_{\mathrm{op}}
  =\norm{[U_T-W,Z_a]}_{\mathrm{op}}
  \le2\norm{U_T-W}_{\mathrm{op}}.
\end{equation}
Thus every operator-norm $1/2$-approximation of $U_T$ must touch all these qubits, hence requires at least $pT/8$ gates.
\end{proof}

\section{Kac's random walk}\label{sec:Kac}

In this section, we study the complexity growth of Kac's random walk on $\mathbf{SO}(N)$ (recall $N=2^n$).
In Kac's random walk, we pick a random pair of coordinates $i,j\in [N]$ and then perform a random rotation in the $i$-$j$ plane \cite{Kac1956}.
We denote by $U_T$ the random variable given by the $T$-step Kac random walk.

The next theorem gives a small-ball estimate with an $\exp(-\Omega(T))$ error term for an exponentially large time range.
The argument is similar to \cref{prop:local-short-complexity}. The key observation is still that each step explores different directions.
\begin{theorem}\label{thm:kac-small-ball}
There exists an absolute constant $c_0>0$ such that, for $2\leq T\le c_0N$, 
\begin{align}
  \sup_{V\in\mathbf{SO}(N)}\prob*{\norm{U_T-V}_{\mathrm{op}}\le \frac12} \le \exp(-\Omega(T)).\label{eq:kac-small-ball}
\end{align}
\end{theorem}

\begin{proof}
  Fix the number of steps as $T$.
Let $N_1$ count the coordinates touched exactly once.
The argument in the proof of \cref{prop:local-short-complexity} applies here, giving, for sufficiently small $c_0$:
\begin{equation}
\prob{N_1<T}\le \exp(-\Omega(T)).
\end{equation}

For any fixed random walk history with $N_1\ge T$, we can choose a set $\mathcal Q\subset [N]$ of at least $T/2$ coordinates whose unique incident gates are all distinct.
For $a\in\mathcal Q$, write $j(a)$ for its unique incident gate and $\theta_{j(a)}$ for that gate's angle.
All other rotations fix coordinate $a$, so $(U_T)_{aa}=\cos\theta_{j(a)}$.

For any fixed $V\in\mathbf{SO}(N)$, the event $\norm{U_T-V}_{\mathrm{op}}\le1/2$ implies $\abs{\cos\theta_{j(a)}-V_{aa}}\le1/2$ for every $a\in\mathcal Q$.
Note that, for any fixed $x\in\mathbb R$, when $\theta$ is chosen uniformly at random from $[0,2\pi)$, we have:
\begin{equation}
  \prob*{\abs{\cos\theta-x}\le \frac12}\le \frac12.
\end{equation}
Moreover, conditional on the coordinates, the angles $\theta_{j(a)}$ are independent for different $a\in\mathcal Q$.
Therefore, 
\begin{equation}
  \prob*{\norm{U_T-V}_{\mathrm{op}}\le\frac12} \le 2^{-\abs{\mathcal Q}}\le2^{-T/2}.
\end{equation}
Adding the probability $\prob{N_1<T}\le e^{-\Omega(T)}$ of the excluded pair histories, we obtain
\begin{equation}
  \prob*{\norm{U_T-V}_{\mathrm{op}}\le\frac12}
  \le 2^{-T/2}+e^{-\Omega(T)}
  =e^{-\Omega(T)}.
\end{equation}
\end{proof}

Following the same $\epsilon$-net argument as in the proof of
\cref{cor:local-complexity-growth}, the above theorem implies
\begin{equation}
    \prob*{\mathcal C_{1/4}(U_T)\le \frac{cT}{\log n+\log T}} \le \exp(-\Omega(T)).\label{eq:kac-complexity}
\end{equation}
Here $\mathcal C_\epsilon(\cdot)$ remains defined as the minimum number of 2-local orthogonal gates (rather than, for example, the minimum number of steps in Kac's walk) required to approximate an element in $\mathbf{SO}(N)$.

On the other hand, it is known that Kac's random walk has a uniform spectral gap $\Omega(1/N)$ \cite{CarlenCarvalhoLoss2003}.
The spectral gap method, similar to that in \cref{sec:spectral-gap-complexity}, then gives
\begin{equation}
  \sup_{V\in\mathbf{SO}(N)}\prob*{\norm{U_T-V}_{\mathrm{op}}\le \frac12} \le \exp(-\Omega(\min\{\frac{T}{N},N^2\})).
\end{equation}
For the time range $T=O(N)$, the above estimate itself does not give a growing lower bound on the complexity.

\section*{Statement on AI use}
This work was developed through collaboration between the author and AI (ChatGPT 5.6 Sol).
The author conceived the overall proof strategy.
The technical lemmas in \cref{sec:all-to-all-local} were refined through extensive discussions with AI. In particular, AI proposed and analyzed the final form of the selectors used in that section, which are crucial for our work.
\appendix

\section{Complexity from spectral gap}\label{sec:spectral-gap-complexity}

In this appendix, we show how spectral gaps imply both small-ball bounds and circuit-complexity growth.
The implication is standard; see, e.g., \cite{BHH2016, BCHKP2021,CHHLMT}.

Recall that $N=2^n$ and $D=N^2-1$.
\begin{lemma}\label{prop:general-gap-complexity-appendix}
  Suppose a random walk has a spectral gap $0<\gamma\leq 1$, uniformly for all $k$-th moments.
% Let $U_T=G_T\cdots G_1$, where the gates are independent with common distribution $\nu$ on $\mathbf U(N)$.
For every integer $T\ge1$, put $L_T=\min\{\gamma T,D\}$.
Then, with an absolute constant $c>0$,
\begin{align}
  &\sup_{V\in\mathbf U(N)}\prob*{d_2(U_T,V)\le1}
  \le \exp(-\Omega(L_T)),\label{eq:general-gap-small-ball-appendix}\\
  &\prob*{\mathcal C_{1/2}(U_T)\le
  \frac{cL_T}{\log n+\log(1+L_T)}}
  \le \exp(-\Omega(L_T)).\label{eq:general-gap-complexity-appendix}
\end{align}
\end{lemma}

\begin{proof}

Denote $\rho_k(U)=U^{\otimes k}\otimes\overline U^{\otimes k}\in \mathrm{End}(\mathcal H^{\otimes k}\otimes \mathcal{\bar H}^{\otimes k})$ where $\mathcal H=(\mathbb C^2)^{\otimes n}$.
$M_k=\EX{\rho_k(G)}$ where the expectation is over the one-step random walk, and $P_k=\EX[\mathrm{Haar}]{\rho_k(U)}$.
The assumption says $\norm{M_k-P_k}_{\mathrm{op}}\le 1-\gamma \leq e^{-\gamma}$ for all $k\ge1$.
Iterating the random walk for $T$ steps, we have
$\norm{M_k^T-P_k}_{\mathrm{op}}\le e^{-\gamma T}$.

Now fix an arbitrary point $V\in \mathbf U(N)$. Let us consider the overlap $\tau(V^\dagger U)$.
% Moreover, $\rho_k(V^\dagger)P_k=P_k$, since $P_k$ projects onto the invariant subspace.
We have:
\begin{align}
  \EX{\abs{\tau(V^\dagger U_T)}^{2k}}
  &=N^{-2k}\Tr_k\!\left(\rho_k(V^\dagger)M_k^T\right)\\
  &=N^{-2k}\Tr_k P_k
    +N^{-2k}\Tr_k\!\left(
      \rho_k(V^\dagger)(M_k^T-P_k)\right)\\
  &\le N^{-2k}\Tr_k P_k
    +\norm{M_k^T-P_k}_{\mathrm{op}}\\
  &\le \frac{k!}{N^{2k}}+e^{-\gamma T}.
\end{align}
Here $\Tr_k$ is the trace on $\operatorname{End}(\mathcal H^{\otimes k}\otimes \mathcal{\bar H}^{\otimes k})$;
the last step uses the following inequality due to the Schur-Weyl duality:
\begin{equation}
  \tr P_k =\dim\operatorname{End}_{\mathbf{U}(N)}(\mathcal H^{\otimes k}) \leq k! \,.
\end{equation}

As $d_2(U_T,V)\le1$ implies $\abs{\tau(V^\dagger U_T)}\ge1/2$, Markov's inequality yields:
\begin{equation}\label{eq:comparison-gap-small-ball}
  \sup_V\prob*{d_2(U_T,V)\le1}
  \le 2^{2k}
  \left(\frac{k!}{N^{2k}}+e^{-\gamma T}\right).
\end{equation}
Take $k=\alpha L_T$ for a sufficiently small absolute $\alpha>0$l, then both terms in the r.h.s. are $\exp(-\Omega(L_T))$.

The circuit net from \cref{eq:local-circuit-net} has logarithmic size at most $Cr[\log n+\log(1+r)]$.
A union bound with \cref{eq:general-gap-small-ball-appendix} proves \cref{eq:general-gap-complexity-appendix}.
\end{proof}

  It is shown in \cite{BaerHaah2026} that the all-to-all 2-local random quantum circuits have a uniform spectral gap $\gamma=\Omega(1/n)$. So we have:
\begin{corollary}\label{thm:gap-complexity-appendix}
For all-to-all 2-local random quantum circuits, with $T$ counting the total number of gates and $L_T=\min\{T/n,D\}$, there exists an absolute constant $c>0$ such that
\begin{align}
  &\sup_{V\in\mathbf U(N)}\prob*{d_2(U_T,V)\le1}
  \le \exp(-\Omega(L_T)),\label{eq:gap-optimized-small-ball-appendix}\\
  &\prob*{\mathcal C_{1/2}(U_T)\le
  \frac{cL_T}{\log n+\log(1+L_T)}}
  \le \exp(-\Omega(L_T)).\label{eq:gap-all-time-complexity-appendix}
\end{align}
\end{corollary}

\bibliographystyle{alpha}
\bibliography{ref}

\end{document}